\documentclass[journal]{IEEEtran}

\usepackage[final]{listings}
\usepackage{acro}

\DeclareAcronym{ASE}{
	short = ASE,
	long= amplifier-induced spontaneous emission
}

\DeclareAcronym{AWGN}{
	short = AWGN,
	long= additive white Gaussian noise
}

\DeclareAcronym{BER}{
	short = BER,
	long= bit error rate
}

\DeclareAcronym{BICM}{
	short = BICM,
	long= bit-interleaved coded modulation
}

\DeclareAcronym{BRGC}{
	short = BRGC,
	long = binary reflected Gray code
}

\DeclareAcronym{CCDM}{
	short = CCDM,
	long = constant composition distribution matcher
}

\DeclareAcronym{CD}{
	short = CD,
	long = chromatic dispersion
}

\DeclareAcronym{C3SE}{
	short = C3SE,
	long = Chalmers Centre for Computational Science and Engineering
}

\DeclareAcronym{COIN}{
	short = COIN,
	long = coding for optical communications in the nonlinear regime
}

\DeclareAcronym{DB}{
	short = DBP,
	long = digital backpropagation
}

\DeclareAcronym{DM}{
	short = DM,
	long = distribution matcher
}

\DeclareAcronym{DVBS2}{
	short = DVB-S2,
	long = digital video broadcasting - satellite 2
}

\DeclareAcronym{EDFA}{
	short = EDFA,
	long = erbium-doped fiber amplifier
}

\DeclareAcronym{FEC}{
	short = FEC,
	long = forward error correction
}

\DeclareAcronym{iid}{
	short = i.i.d.,
	long= independently and identically distributed
}

\DeclareAcronym{INFT}{
	short = INFT,
	long= inverse nonlinear Fourier transform
}

\DeclareAcronym{LDPC}{
	short = LDPC,
	long = low-density parity-check
}

\DeclareAcronym{MI}{
	short = MI,
	long = mutual information
}

\DeclareAcronym{NFDM}{
	short = NFDM,
	long= nonlinear frequency-division multiplexing
}

\DeclareAcronym{NFT}{
	short = NFT,
	long= nonlinear Fourier transform,
	alt = Nonlinear Fourier Transform
}

\DeclareAcronym{NLSE}{
	short = NLSE,
	long= nonlinear Schr{\"o}dinger equation
}

\DeclareAcronym{OOK}{
	short = OOK,
	long= on-off keying
}

\DeclareAcronym{PAS}{
	short = PAS,
	long= probabilistic amplitude shaping
}

\DeclareAcronym{PDF}{
	short = PDF,
	long= probability density function
}

\DeclareAcronym{PES}{
	short = PES,
	long= probabilistic eigenvalue shaping
}

\DeclareAcronym{PMF}{
	short = PMF,
	long= probability mass function
}

\DeclareAcronym{PS}{
	short = PS,
	long= probabilistic shaping
}

\DeclareAcronym{rv}{
	short = RV,
	long= random variable
}

\DeclareAcronym{SMF}{
	short = SMF,
	long= single mode fiber
}

\DeclareAcronym{SNR}{
	short = SNR,
	long= signal-to-noise ratio
}
\DeclareAcronym{SSMF}{
	short = SSMF,
	long= standard single mode fiber
}

\DeclareAcronym{SSFM}{
	short = SSF, 
	long= split-step Fourier
}

\usepackage{xifthen}
\usepackage{cite}
\usepackage[colorlinks=false,
						bookmarks=true,
						citecolor=blue,
						urlcolor=blue]
						{hyperref} 

\usepackage{amsmath,amssymb,amsfonts,amsthm} %
\usepackage{bm}
\usepackage{siunitx}
\usepackage{ctable}
\usepackage{xcolor}
\usepackage{algorithm,algorithmic}
\usepackage{ifthen}

\usepackage{marginnote}

\definecolor{curve_color01}{HTML}{000000}
\definecolor{curve_color02}{RGB}{239,138,98}
\definecolor{curve_color03}{RGB}{103,169,207}
\definecolor{curve_color04}{RGB}{33,102,172}
\definecolor{curve_color05}{RGB}{178,24,43}
\definecolor{curve_color06}{RGB}{253,219,199}

\definecolor{color_weak_black}{RGB}{20,20,20}

\definecolor{block_gray}{RGB}{204,204,204}
\definecolor{block_background_gray}{RGB}{240,240,240}
\definecolor{block_background_line_gray}{RGB}{175,175,175}

\newcommand{\internalLink}[1]{\hyperref[#1]{\ref*{#1}}}
\newcommand{\internalLinkChapter}[1]{Chapter \hyperref[#1]{\ref*{#1}}}
\newcommand{\internalLinkSection}[1]{Section \hyperref[#1]{\ref*{#1}}}
\newcommand{\internalLinkAppendix}[1]{Appendix \hyperref[#1]{\ref*{#1}}}
\newcommand{\internalTab}[1]{Table\,\hyperref[#1]{\ref*{#1}}}
\newcommand{\internalFig}[1]{Fig.\,\hyperref[#1]{\ref*{#1}}}
\newcommand{\internalEq}[1]{(\hyperref[#1]{\ref*{#1}})}
\newcommand{\internalEqs}[2]{(\hyperref[#1]{\ref*{#1}})-(\hyperref[#2]{\ref*{#2}})}
\newcommand{\internalListing}[1]{Listing \hyperref[#1]{\ref*{#1}}}
\newcommand{\internalAlgorithm}[1]{Algorithm\,\hyperref[#1]{\ref*{#1}}}
\newcommand{\internalLemma}[1]{Lemma \hyperref[#1]{\ref*{#1}}}

\newcommand{\internalDefinition}[1]{Def. \hyperref[#1]{\ref*{#1}}}

\newtheorem{definition}{Definition}
\newtheorem{example}{Example}
\newtheorem{lemma}{Lemma}
\newtheorem{theorem}{Theorem}

\newcommand{\besselI}[2]{\operatorname{I}_{#1}\left(#2\right)}
\renewcommand{\exp}[1]{\operatorname{e}^{#1}}

\newcommand{\sech}[1]{\operatorname{sech}\left(#1\right)}

\newcommand{\RealNumbers}[0]{\mathbb{R}}
\newcommand{\ComplexNumbers}[0]{\mathbb{C}}

\newcommand{\E}[2][]{%
\mathbb{E}_{#1}\!\left\{{#2}\right\}
}

\renewcommand{\Im}[1]{\mathfrak{I}\{#1\}}
\renewcommand{\Re}[1]{\mathfrak{R}\{#1\}}

\newcommand{\MI}[1]{\mathbb{I}(#1)}

\newcommand{\MIstart}[1]{\mathsf{I}(#1)}

\newcommand{\Istar}[1]{\mathsf{I}(#1)}
\newcommand{\Cstar}{\mathsf{C}}

\newcommand{\entropy}[1]{\mathbb{H}(#1)}

\newcommand{\transpose}[1]{#1^\mathsf{T}}

\newcommand{\T}[1]{T\left(#1\right)}

\DeclareSIUnit{\belmilliwatt}{Bm}
\DeclareSIUnit{\dBm}{\deci\belmilliwatt}

\makeatletter
\let\MYcaption\@makecaption
\makeatother

\usepackage[font=footnotesize]{subcaption}

\makeatletter
\let\@makecaption\MYcaption
\makeatother

\let\originalleft\left
\let\originalright\right
\renewcommand{\left}{\mathopen{}\mathclose\bgroup\originalleft}
\renewcommand{\right}{\aftergroup\egroup\originalright}

\begin{document}

\title{Probabilistic Eigenvalue Shaping for\\ Nonlinear Fourier Transform Transmission}

\author{Andreas~Buchberger, Alexandre~Graell~i~Amat,~\IEEEmembership{Senior~Member,~IEEE,}\\ Vahid~Aref,~\IEEEmembership{Member,~IEEE,} and~Laurent~Schmalen,~\IEEEmembership{Senior~Member,~IEEE}
\thanks{This work was funded by the European Union's Horizon 2020 research and innovation programme under the Marie Sk\l{}odowska-Curie grant agreement No. 676448.}%
\thanks{A. Buchberger is with the Department of Electrical Engineering, Chalmers University of Technology, Gothenburg, SE-412 96,
Sweden and Nokia Bell Labs, Lorenzstr. 10, 70435 Stuttgart, Germany, e-mail: andreas.buchberger@chalmers.se.}
\thanks{A. Graell i Amat is with the Department of Electrical Engineering, Chalmers University of Technology, Gothenburg, SE-412 96,
Sweden, e-mail: alexandre.graell@chalmers.se.}
\thanks{V. Aref and L. Schmalen are with Nokia Bell Labs, Lorenzstr. 10, 70435 Stuttgart, Germany, e-mail: \{firstname.lastname\}@nokia-bell-labs.com.}
}

\maketitle

\begin{abstract}
We consider a \ac{NFT}-based transmission scheme, where data is embedded into the imaginary part of the nonlinear discrete spectrum. Inspired by probabilistic amplitude shaping, we propose a \ac{PES} scheme as a means to increase the data rate of the system. We exploit the fact that for an \ac{NFT}-based transmission scheme, the pulses in the time domain are of unequal duration by transmitting them with a dynamic symbol interval and find a capacity-achieving distribution. The \ac{PES} scheme shapes the information symbols according to the capacity-achieving distribution and transmits them together with the parity symbols at the output of a \acl{LDPC} encoder, suitably modulated, via time-sharing. We furthermore derive an achievable rate for the proposed \ac{PES} scheme. We verify our results with simulations of the discrete-time model as well as with \acl{SSFM} simulations.
\end{abstract}


\IEEEpeerreviewmaketitle

\acresetall


\section{Introduction}
\label{sec:introduction}
\IEEEPARstart{P}{ulse} propagation in optical fibers is severely impaired by nonlinear effects that should be either
compensated or utilized for the design of the communication system.
The \ac{NFT}~\cite{yousefi2014nftI-III} provides a method to transform a signal from the time domain into a nonlinear frequency domain (spectrum), where the channel acts as a multiplicative filter on the signal. 
 The nonlinear spectrum consists of a continuous and a discrete part.
 Both parts can be used to transmit information, either separately or jointly, and several schemes have been presented in theory and practice~\cite{yousefi2014nftI-III,dong2015,aref2015ecoc,aref2016ecoc,geisler2016ecoc, hari2016}.
However, very little is known so far about the \ac{PDF} of the received signal in the nonlinear spectral domain when it is contaminated by channel noise. 
   
In~\cite{shevchenko2018tcom}, a simplified communication system  modulating only the imaginary part 
 of the eigenvalues in the discrete nonlinear spectrum was presented. For this scheme, an approximation for the conditional \ac{PDF} of the channel can be obtained in closed form.
In general, for a given channel, the capacity-achieving distribution is not known and is often different from the conventional distribution with equispaced signal points and uniform signaling. Hence, some form of shaping is required~\cite{forney1984}.
Two popular methods of shaping are probabilistic shaping and geometric shaping.
In geometric shaping, the capacity-achieving distribution is mimicked by optimizing the position of the constellation points for equiprobable signaling~\cite{sun1993_geometric_shaping} whereas probabilistic shaping uses uniformly spaced constellation points and approximates the capacity-achieving distribution by assigning different probabilities to different constellation points~\cite{forney1984}.

The main drawback of probabilistic shaping is its practical implementation. An abundance of probabilistic shaping schemes have been presented, most suffering from high decoding complexity, low flexibility in adapting the spectral efficiency, or error propagation. For a literature review on probabilistic shaping, we refer the reader to~\cite[Section II]{bocherer2015_bw_efficient_rate_matched_ldpc}.

Recently, a new scheme called \ac{PAS} has been proposed  in~\cite{bocherer2015_bw_efficient_rate_matched_ldpc}. Compared to other shaping schemes, \ac{PAS} yields high flexibility and close-to-capacity performance over a wide range of spectral efficiencies for the \ac{AWGN} channel while still allowing bit-metric decoding. Although originally introduced for the \ac{AWGN} channel, \ac{PAS} can be applied to other channels with a symmetric capacity-achieving input distribution assuming a sufficiently high spectral efficiency.

In this paper, we consider a similar \ac{NFT}-based transmission scheme to the one presented in~\cite{shevchenko2018tcom}, where data is embedded into the imaginary part of the nonlinear discrete spectrum. As a means to increase the data rate, we demonstrate that the concept of \ac{PAS} can be adapted to this \ac{NFT}-based transmission system. In particular, we propose a \ac{PES} scheme, enabling similar low complexity and bit-metric decoding as \ac{PAS}.
We take advantage of the dependence of the pulse length on the data  for the \ac{NFT}-based transmission system and transmit each pulse as soon as the previous one has been transmitted rather than with a fixed interval as in~\cite{shevchenko2018tcom}, yielding increased data rate.
Accordingly, we find the capacity-achieving input distribution, maximizing the time-scaled \ac{MI}. For ease of notation, we refer to the maximized \ac{MI} as capacity noting that it is in fact the constrained capacity of a system transmitting first-order solitons.
The \ac{PES} scheme then shapes the information symbols according to the capacity-achieving distribution by a \ac{DM}. The information symbols are also encoded by a \ac{LDPC} encoder and the parity symbols at the output of the encoder are suitably modulated. The resulting sequence of modulated symbols and the sequence at the output of the \ac{DM} are transmitted via time-sharing. We further  derive an achievable rate for such a \ac{PES} scheme. 
We demonstrate via discrete-time Monte-Carlo and \ac{SSFM} simulations, that \ac{PES} performs at around \(\SI{2}{\decibel}\) from capacity using off-the-shelf \ac{LDPC} codes. The proposed \ac{PES} scheme yields a significant improvement of up to twice the data rate compared to an unshaped system as in~\cite{shevchenko2018tcom}.

It is important to note that although first-order solitons do not outperform conventional coherent systems due to their spectrally inefficient pulse shape compared to a Nyquist pulse shape, they have some other advantages. For instance, the first-order soliton transmission does not require \ac{CD} compensation or \ac{DB} as dispersion and nonlinearity are balanced and hence compensated.
This work attempts to approach the limits of current NFT-based systems. 
To improve the spectral efficiency further, one should use higher-order solitons as well as the continuous part of the nonlinear spectrum together~\cite{aref2016ecoc}. However, the channel equalization will not be as easy as the one for the first-order solitons and the channel model is not yet fully known.

The remainder of the paper is organized as follows. In \internalLinkSection{sec:system}, we describe pulse propagation in an optical fiber and the \ac{NFT}-based transmission scheme. In \internalLinkSection{sec:optimum_input_distribution}, we optimize the input distribution and in \internalLinkSection{sec:shaping_time_share}, we introduce and describe the proposed \ac{PES} scheme and derive an achievable rate. In \internalLinkSection{sec:performance}, we present numerical results for \ac{PES}, both from Monte-Carlo simulation and \ac{SSFM} simulation, and in \internalLinkSection{sec:conclusion} we draw some conclusions.

Notation: The following notation is used throughout the paper. \(\Re{\cdot}\) and \(\Im{\cdot}\) denote the real and the imaginary part of a complex number, respectively, and \(\jmath=\sqrt{-1}\) denotes  the imaginary unit. Vectors are typeset in bold, e.g., \(\bm{x}\), \acp{rv} are capitalized, e.g., \(X\), and hence vectors of \acp{rv} are capitalized bold, e.g., \(\bm{X}\). The \ac{PDF} of an \ac{rv} \(X\) is written as \(p_X(x)\) and its expectation as \(\E[X]{x}\). The conditional \ac{PDF} of \(Y\) given \(X\) is denoted as \(p_{Y|X}(y|x)\). The \ac{PMF} of an \ac{rv} \(X\) is denoted by \(P_X(x)\). The transpose of a vector or matrix is given as \((\cdot )^\mathsf{T}\). A set is denoted by a capitalized Greek letter, e.g., \(\Lambda\), and its cardinality by \(|\Lambda|\). We write \(\log_a(\cdot)\) for the logarithm of base \(a\) and \(\ln(\cdot)\) for the natural logarithm.


\section{Nonlinear Fourier Transform-based Transmission System}
\label{sec:system}
\subsection{Pulse Propagation and the Nonlinear Fourier Transform}
\label{subsec:NFT}
Pulse propagation in optical fibers is governed by  a partial differential equation, the stochastic \ac{NLSE},
\begin{multline}
		\jmath\frac{\partial u(\tau,\ell)}{\partial \ell} + \jmath\frac{\alpha}{2}u(\tau,\ell) - \frac{\beta_2}{2}\frac{\partial^2 u(\tau,\ell)}{\partial \tau^2} + \gamma u(\tau,\ell) |u(\tau,\ell)|^2  \\= n(\tau,\ell ),\nonumber
\end{multline}	
where \(u(\tau,\ell)\) denotes the envelope of the electrical field as a function of the position \(\ell\) along the fiber and time \(\tau\), \(\alpha\) the attenuation, \(\beta_2\) the second order dispersion, \(\gamma\) the nonlinearity parameter, and \(n(\tau, \ell)\) 
is a white Gaussian process in time and in space with spectral density \(\sigma_0^2\). The spectral density depends on the system and  for distributed Raman amplification is given as \(\sigma_0^2 = \alpha K_T h \nu_0\), where \(K_T\) is the temperature-dependent phonon occupancy factor, and \(h\nu_0\) is the average photon energy~\cite{shevchenko2018tcom}.
A general closed-form solution of the stochastic \ac{NLSE} does not exist. In some special cases, e.g., for noisefree and lossless fibers, special solutions like, e.g., solitons, exist. Furthermore,  we consider the \ac{NLSE} in normalized form in the focusing regime, i.e., \(\beta_2 < 0\), under the assumption of ideal distributed Raman amplification, i.e., \(\alpha = 0\),
\begin{align}
	\jmath\frac{\partial q(t,z)}{\partial z}  + \frac{\partial^2 q(t,z)}{\partial t^2} + 2 q(t,z) |q(t,z)|^2  =  0,
	\label{eqn:nlse_losless_noisefree}
\end{align}
where \(t=\tau / \sqrt{|\beta_2| L/2}\), \(z=\ell / L\), \(q = u\sqrt{\gamma L}/\sqrt{2}\), and \(L\) is the length of the fiber.
In this case, the \ac{NLSE} is an integrable partial differential equation for which a pair of operators, called Lax pair, can be found. The eigenvalues of such an operator remain invariant during noiseless propagation and the Lax pair can be used to solve the partial differential equation. 
Solutions of \internalEq{eqn:nlse_losless_noisefree} can be uniquely represented in terms of its eigenvalues via the so-called \ac{NFT}.
For a given position \(z\), the \ac{NFT} of a signal \(q(t)\) (we drop the position \(z\) for simplicity of presentation) 
with support on the time interval \(t\in [t_1,t_2]\), is calculated by solving the partial differential equation
\begin{align}
\frac{\partial \bm{v}(t,\lambda)}{\partial t} = \begin{pmatrix}
-\jmath\lambda & q(t)\\ 
- q(t)^*&  \jmath\lambda
\end{pmatrix}\bm{v}(t,\lambda),\label{eqn:nft_dgl}
\end{align}
where \(\bm{v}(t,\lambda)=\begin{pmatrix} v_1(t,\lambda) & v_2(t,\lambda) \end{pmatrix}\) is the eigenvector of the auxiliary operator,
with boundary conditions
\begin{align*}
\bm{v}^{(1)}(t,\lambda) &\rightarrow \transpose{\begin{pmatrix} 0 & 1 \end{pmatrix}} \exp{\jmath\lambda t}, &\text{as } t &\rightarrow t_2\\
\bm{v}^{(2)}(t,\lambda) &\rightarrow \transpose{\begin{pmatrix} 1 & 0 \end{pmatrix}} \exp{-\jmath\lambda t}, &\text{as } t &\rightarrow t_1,
\end{align*}
 and \(\lambda\) is the spectral component.
Solving \internalEq{eqn:nft_dgl} gives rise to the continuous and discrete nonlinear spectrum
\begin{align*}
	\hat{q}(\lambda ) &= \frac{b(\lambda)}{a(\lambda)}, \lambda\in\RealNumbers & \tilde{q}(\lambda_i) &= \frac{b(\lambda_i)}{\mathrm{d}a(\lambda) / \mathrm{d}\lambda|_{\lambda=\lambda_i}}, \lambda_i\in\ComplexNumbers^{+},
\end{align*}
respectively, where \(a(\lambda ) = \lim_{t\rightarrow t_2} v_1^{(2)}(t,\lambda)	\exp{\jmath\lambda t}\), \(b(\lambda ) = \lim_{t\rightarrow t_2} v_2^{(2)}(t,\lambda)	\exp{-\jmath\lambda t}\), and \(\lambda_i\) are the zeros of \(a(\lambda)\), \(\lambda_i\in\ComplexNumbers^{+}\), a finite set of isolated complex zeros, referred to as eigenvalues.
Hence, the \ac{NFT} represents the signal in the nonlinear spectral domain, where the influence of the channel on the signal is a multiplicative filter.

As a counterpart to the \ac{NFT} that transforms a signal from the time domain to the nonlinear spectral domain, the \ac{INFT} transforms a signal from the nonlinear spectral domain to the time domain. For an in-depth mathematical description of the \ac{INFT}, we refer the interested reader to~\cite{yousefi2014nftI-III}.

\subsection{Soliton Transmission}
\begin{figure}[!t]
  \centering
		\includegraphics{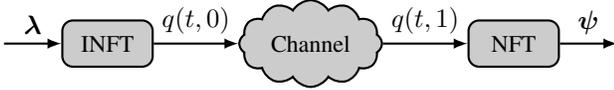}
		\caption{Block diagram of the \ac{NFT}-based system.}
		\label{fig:system:block_diagram_nft}
\end{figure}

As in~\cite{shevchenko2018tcom}, we embed information in the imaginary part of the discrete spectrum, also referred to as eigenvalues. 
Hence, the input of the channel is an \ac{rv} \(X\in\Lambda=\{\lambda_1,\ldots,\lambda_M\}\), where \(\Lambda\) is the set of eigenvalues, \(\lambda_i\) is the \(i\)th eigenvalue,  and \(M\) is the order of the modulation. The eigenvalues \(\{\lambda_i\}\) are assumed to be ordered in ascending order by their imaginary parts. Furthermore, the output of the channel is an \ac{rv} \(Y\in\Psi\), where \(\Psi = \{y\in\mathbb{C}: \Re{y}=0, \Im{y}\geq 0\}\).
A block diagram is depicted in \internalFig{fig:system:block_diagram_nft}. The information embedded in a single eigenvalue \(\lambda\in\Lambda\) is transformed to a time-domain signal \(q(t,0)\) via the \ac{INFT} where the transmitter is located at position \(z=0\) along the fiber. At position \(z=1\), the receiver calculates the discrete spectrum \(\psi\in\Psi\) from the received signal \(q(t,1)\) via the \ac{NFT}.
The time-domain signal corresponds to first order solitons, i.e., 
\begin{align*}
	q(t,0) &= 2\Im{\lambda} \sech{2\Im{\lambda} t}.
\end{align*}
For the \ac{NFT} to be valid, the signal must have finite support, i.e., before transmitting the next pulse, the previous one must have returned to zero. As the pulses in general have infinite tails, we truncate them when they fall below a threshold close to zero. We define the pulse over the smallest support containing a fraction \((1-\delta ) \) of the energy of the pulse and hence, we can formally define the pulse width as follows.

\begin{definition}
The pulse width of \(\lambda\) is defined as the smallest support containing a fraction \((1-\delta ) \) of the energy of the pulse, 
\begin{align*}
	T(\lambda, \delta ) &\triangleq \frac{1}{2\Im{\lambda}} \ln\left(\frac{2}{\delta} -1\right),
\end{align*}
where \(0 < \delta < 1\).
\end{definition}

The value of the cutoff parameter \(\delta\) must be chosen in a way such that soliton-soliton interactions are negligible. For longer transmission distances, \(\delta\) decreases, i.e., the pulses must be spaced further apart. Furthermore, the condition
\begin{align}
	\exp{-2\Im{\lambda} \Delta(\lambda, \delta )} 
	&= \exp{-\ln\left(\frac{2}{\delta} -1\right)} \ll 1
	\label{eqn:guard_time_condition}
\end{align}
must be fulfilled~\cite{shevchenko2018tcom}.

At this point, it is important to comment on the memorylessness of the system emanating from the absence of soliton-soliton interactions. 
A pulse train of well-separated first order solitons was investigated in~\cite{shevchenko2018tcom} for launch powers of \(\SI{-1.5}{\dBm}\) and \(\SI{1.45}{\dBm}\) and transmission over \(\SI{500}{\kilo\meter}\) and \(\SI{2000}{\kilo\meter}\). It was shown via \ac{SSFM} simulations that the correlation between the symbols at the receiver is essentially zero, concluding that the channel is indeed memoryless in the transmission range of \(\SI{500}{\kilo\meter}\) to \(\SI{2000}{\kilo\meter}\) and transmit power range of \(\SI{-1.5}{\dBm}\) to \(\SI{1.45}{\dBm}\) for which the model \internalEq{eqn:pdf_yx_sac} is applicable.
While this approach is not a rigorous proof, the results indicate that memorylessness is a valid assumption. Although the transmission scheme is different in~\cite{shevchenko2018tcom}, the underlying condition that any two pulses need to be sufficiently separated is the same. Hence, we can treat the \ac{NFT}-based transmission system in this work as a memoryless channel.

In a practical system, we assume distributed Raman amplification and \ac{ASE} noise with received power spectral density \(\sigma^2\) to compensate for the lossy fiber and be able to use the \ac{NFT} to relate the input and the output. The conditional \ac{PDF} of such a system has been derived via a perturbative approach and the Fokker-Planck equation method~\cite{derevyanko2005} and is used to design a communication system in~\cite{shevchenko2018tcom}. It is given by

\begin{align}
	p_{Y|X}(\psi|\lambda) &= \frac{2}{\sigma^2}\sqrt{\frac{\Im{\psi}}{\Im{\lambda}}}\exp{-2\frac{\Im{\lambda}+\Im{\psi}}{\sigma^2}}\besselI{1}{\frac{4\sqrt{\Im{\lambda}\Im{\psi}}}{\sigma^2}},
\label{eqn:pdf_yx_sac}
\end{align}
where \(\psi\) is the received symbol as in \internalFig{fig:system:block_diagram_nft}, and \(\besselI{1}{\cdot}\) is the modified Bessel function of the first kind of order one. The power spectral density of the received \ac{ASE} noise \(\sigma^2\) is normalized and relates to real world units as \(\sigma^2 = \gamma\sqrt{L^3}\sigma_0^2 / \left(\sqrt{2|\beta_2|}\right)\).
The \ac{SNR} is defined as \(\mathrm{SNR}\triangleq 4\E[X]{\Im{\lambda}}/\sigma^2\).
It is important to note that the model \internalEq{eqn:pdf_yx_sac} assumes the noise intensity to be small such that it can be treated as a perturbation to the soliton. Hence, the model is only applicable if the signal energy is not the same order as that of the noise. Furthermore, for very high signal powers, \internalEq{eqn:pdf_yx_sac} is no longer valid either since the impact of the inelastic scattering effects (i.e., stimulated Raman or Brillouin scattering) is not considered within the 1st-order perturbation approach. For a detailed derivation of the model, we refer the reader to~\cite{derevyanko2005}.

In~\cite{shevchenko2018tcom}, the shortest possible symbol interval is defined by the pulse duration of \(\lambda_1\), i.e., the longest pulse. However, this tends to be inefficient since especially for short pulses, the guard interval between two consecutive pulses is longer than necessary and thereby limits the data rate.
 Here, we exploit the effect of varying pulse lengths and transmit each pulse as soon as the previous one has returned to zero. This concept is depicted in \internalFig{fig:system:plot_comparison_spacing}, where pulse sequences with fixed and varying symbol interval are compared. The figure clearly shows the advantage of a varying pulse interval and also demonstrates the aforementioned inefficiencies. The data rate of a system with varying symbol intervals depends on the distribution of the data. Thus, we define the average symbol interval as follows.
\begin{figure}[tb]
  \centering
		\includegraphics{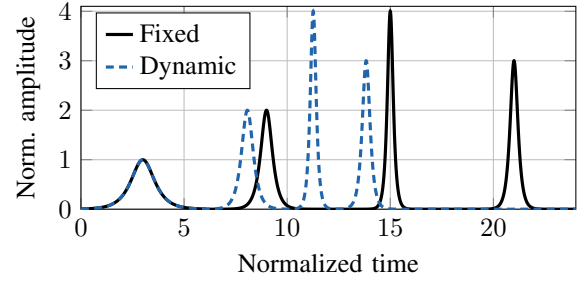}
		\caption{Comparison of a pulse sequence with static symbol intervals and dynamic symbol intervals.}
		\label{fig:system:plot_comparison_spacing}
\end{figure}

\begin{definition}
\label{def:avg_duration}
The average symbol interval is 
\begin{align*}
\bar{T}(X) &\triangleq \sum_{k=1}^{M} p_X (\lambda_k) \T{\lambda_k} = \E[X]{\T{\lambda}}.
\end{align*}
\end{definition}
In~\cite{shevchenko2018tcom}, only eigenvalues with an imaginary part larger than zero are used. We extend this by allowing \(\Im{\lambda}=0\). In the time domain, this results in a pulse with amplitude zero, i.e., we do not transmit anything. We define its corresponding duration as the same as the duration of the shortest pulse, \(\T{\lambda=0} \triangleq \T{\lambda_M}\).

As any practical system can handle only a maximum peak power and a maximum bandwidth, we enforce a peak
power constraint which relates to a maximum eigenvalue constraint. Especially in systems with lumped amplification
and \acp{EDFA}, such a constraint is required as eigenvalues fluctuate depending on their
amplitude, which decreases the performance~\cite{zafrulla2003}.

We note that the varying symbol interval introduces additional challenges on detection. In particular, an erroneously detected symbol may lead to error propagation, insertion errors (detection of symbols when none was transmitted), deletion errors (not detecting a transmitted symbol), or the loss of synchronization. To calculate the capacity, however, we neglect these effects. Hence, the results can be seen as an upper bound on the performance.

\section{Capacity Achieving Distribution}
\label{sec:optimum_input_distribution}
From \internalFig{fig:system:plot_comparison_spacing}, it is intuitive that pulses with short duration should be transmitted more frequently than pulses with long duration. However, shorter pulses are more perturbed by noise than longer pulses. Hence, the optimal input distribution to the channel as described by the conditional \ac{PDF} \internalEq{eqn:pdf_yx_sac} is not the conventional uniform distribution.
The channel capacity is obtained by maximizing the \ac{MI},
\begin{align*}
	\MI{X;Y} &\triangleq \E[X,Y]{\log_2\left({\frac{p_{Y|X}(Y|X)}{\sum_{\tilde{\lambda}\in\Lambda}p_{Y|X}(Y|\tilde{\lambda}) p_X(\tilde{\lambda})}}\right)}
\end{align*}
 over all possible input distributions \(p_X(\lambda )\).
Here, due to the variable transmission duration, we need to consider the \ac{MI} under a variable cost constraint \(\bar{T}(\cdot)\)~\cite{verdu1990}, 
\begin{align}
	\Istar{X;Y} &\triangleq \frac{\MI{X;Y}}{\bar{T}(X)}.
	\label{eqn:mutual_information_time_scaled}
\end{align}
To emphasize that the cost of a symbol is its corresponding pulse duration, we refer to the \ac{MI} in the form of \internalEq{eqn:mutual_information_time_scaled} as time-scaled \ac{MI}. We can therefore define the capacity as
\begin{align}
\Cstar &\triangleq \sup_{p_X(\lambda)} \Istar{X;Y}
\label{eqn:def_cstar}
\end{align}
where we set the supremum to zero if the set of distributions therein is empty.
The capacity-achieving distribution, denoted by \(p_X^*(\lambda )\), is in the set for which the supremum is non-zero.

As the \ac{MI} \(\MI{X;Y}\) is concave in \(p_X(\lambda )\) and \(\bar{T}(X)\) is linear in \(p_X(\lambda )\) and positive, the time-scaled \ac{MI} \(\Istar{X;Y}\) is quasiconcave~\cite[Table 2.5.2]{stancu-minasian_frac_programming}. We can solve \internalEq{eqn:def_cstar} and obtain the corresponding capacity-achieving distribution numerically.
\begin{figure}[!t]
  \centering
		\includegraphics{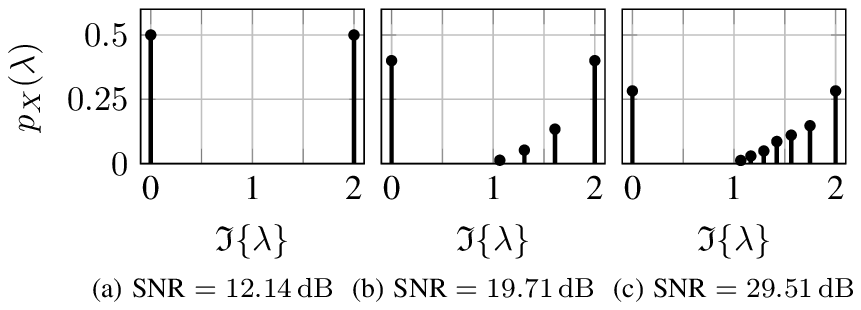}
		\vskip -0.5cm
		\caption{Optimal distribution for different \acp{SNR}.}
		\label{fig:optDist:example_opt_dist}
\end{figure}
Exemplary results of the capacity-achieving distribution are shown in \internalFig{fig:optDist:example_opt_dist}.
We note that the lowest and highest amplitudes are always used with equal and high probability. For low \acp{SNR}, only these are used, i.e., \ac{OOK} is optimal. Furthermore, the capacity-achieving distribution is discrete and is of exponential-like shape with the exception of a point mass at zero as it can be seen in \internalFig{fig:optDist:example_opt_dist}.

Note that \(\Cstar\) assumes memorylessness, which does not necessarily hold due to the variable symbol interval. Hence, \(\Cstar\) is, in fact, the constraint capacity under the assumption of a memoryless channel and the constraint of transmitting only first-order solitons. However, for notational simplicity, we refer to it simply as capacity with its corresponding capacity-achieving distribution.

In the case of a noiseless channel, it is possible to derive a closed form solution to \internalEq{eqn:def_cstar} under the assumption of a finite discretization.
\begin{lemma}
\label{lemma:noiseless}
Let \(\lambda_1, \lambda_2, \ldots, \lambda_M\) be \(M \geq 2\) eigenvalues with \(0 \leq \Im{\lambda_1} < \Im{\lambda_2} < \ldots < \Im{\lambda_M}\) and let \(\T{\lambda_k}\) be the time of transmitting a pulse with eigenvalue \(\lambda_k\). Let $r$ be the unique real positive root of the polynomial $\sum_{k=1}^Mx^{-T(\lambda_k)} - 1$. Then, in the noiseless case, the capacity is obtained as
\[
\Cstar = \log_2(r)
\]
and the capacity-achieving distribution is given by
\begin{align}
 P_X^\diamond(\lambda_k) &= \exp{-\ln(r)\T{\lambda_k}}, & k&=1,\ldots,M.
\label{eqn:opt_dist_noiseless}
\end{align}
\end{lemma}
\begin{proof}
Suppose that the \(k\)-th eigenvalue is transmitted with probability \(P_k\). For any fixed average symbol interval \(\bar{T}(X) = \sum_kP_kT(\lambda_k)\), where \(\T{\lambda_M} \leq \bar{T}(X) \leq \T{\lambda_1}\), we are interested in the distribution that maximizes the entropy while leading to the average symbol duration \(\bar{T}(X)\).
It is known that this distribution takes the form~\cite[Ch. 12]{cover_inf_theory}
\begin{align}
	P_k &= \frac{\exp{-\theta \T{\lambda_k}}}{\xi(\theta)}
	\label{eqn:maximizer}
\end{align}
where \(\xi(\theta) = \sum_i\exp{-\theta T(\lambda_i)}\) ensures that \(\sum_k P_k = 1\)  and $\theta$ has to be selected such that $\sum_kP_kT(\lambda_k) = \bar{T}$. In the noiseless case, the MI is given by $\MI{X;Y} = \entropy{X}$.
The entropy $\entropy{X}$ then is
\begin{align*}
	\entropy{X} =: \entropy{\theta} &= -\sum_{k=1}^M P_k \log_2\left(\frac{\exp{-\theta \T{\lambda_k}}}{\xi(\theta)}\right)\\
									 &= \frac{1}{{\ln(2)}}\sum_{k=1}^M P_k (\theta \T{\lambda_k} + \ln(\xi(\theta)))\\
									 &= \frac{\theta \bar{T}(X)}{{\ln(2)}} + \log_2(\xi(\theta)).
\end{align*}
The time-scaled MI hence takes the form
\begin{align*}
 \MIstart{X;Y} &=  \frac{\theta}{\ln(2)} + \frac{\log_2(\xi(\theta))}{\bar{T}(X)} \\
 &= \frac{\theta}{\ln(2)} + \frac{\log_2(\xi(\theta))}{\sum_kP_kT(\lambda_k)} \\
 &= \frac{\theta}{\ln(2)} + \frac{\log_2(\sum_k\exp{-\theta T(\lambda_k)})\sum_k\exp{-\theta T(\lambda_k)}}{\sum_k\exp{-\theta T(\lambda_k)}T(\lambda_k)}\,. 
\end{align*}
In order to maximize $\MIstart{X;Y}$, we find the optimal parameter $\theta$ by setting $\xi(\theta) = 1$. This can be seen by setting the derivative of $\MIstart{X;Y}$ to zero, with
\begin{multline*}
\frac{\partial}{\partial\theta} \MIstart{X;Y} = \\
 \log_2\left(\sum_k\exp{-\theta T(\lambda_k)}\right)\left(\frac{\sum_k\exp{-\theta T(\lambda_k)}}{\sum_kT(\lambda_k)\exp{-\theta T(\lambda_k)}}\right)^2\mathop{\mathrm{var}}(T(\lambda))\, ,
\end{multline*}
where \(\mathop{\mathrm{var}}(T(\lambda))\) denotes the variance of \(T(\lambda_k)\) for the given \(\theta\).
By assumption, as all $T(\lambda_k)$ are different, the middle part of this expression is strictly positive and $\mathop{\mathrm{var}}(T(\lambda)) > 0$. Hence, it is easy to see that this derivative can only be zero if $\sum_k\exp{-\theta T(\lambda_k)} = 1$. The optimal $\theta$ is hence found by setting $\xi(\hat{\theta}) = 1$. Consider the polynomial
\[
f(x) = \sum_{k=1}^Mx^{-T(\lambda_k)} - 1.
\]
As this polynomial is monotonically decreasing for positive $x$, with $\lim_{x\to 0^+}f(x) = +\infty$ and $\lim_{x\to+\infty}f(x)=-1$, $f(x)$ has exactly one positive real root. Let $r$ be the unique positive real root of $f(x)$. Then $\hat{\theta} = \ln(r)$. Inserting $\hat{\theta}$ into $\MIstart{X;Y}$ and~\eqref{eqn:maximizer} proves the lemma.
\end{proof}

\begin{figure}[!t]
  \centering
		\includegraphics{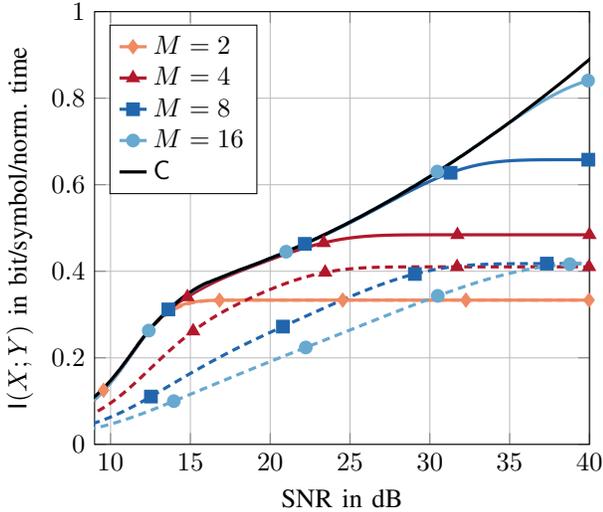}
		\caption{Time-scaled \ac{MI} of the optimal distribution for linearly spaced constellations with \(M\) points (colored with markers solid), and of a system as in~\cite{shevchenko2018tcom} (with markers dotted). As a reference, the capacity \(\Cstar\) is plotted as well (black solid without markers). For the cutoff parameter, \(\delta=0.005\) was used.}
		\label{fig:shaping:mi_comparison}
\end{figure}

We clearly see that \internalEq{eqn:opt_dist_noiseless} is of exponential shape with an additional point mass at zero. Furthermore, we note that the shape of the distribution is mostly caused by the variable pulse duration. The noise then determines the optimal location and optimal number of constellation points. 

For a transmission system, the \ac{MI} is an upper bound on the achievable rate. In \internalFig{fig:shaping:mi_comparison} we evaluate the time-scaled \ac{MI} for various input distributions for a cutoff parameter \(\delta=0.005\). The capacity is depicted with a black solid line.
To reduce the complexity of implementation, we constrain the constellation \(\Lambda\) to \(M\) linearly spaced points from \(\lambda_1=0\) to \(\lambda_M\), i.e.,
\begin{align*}
	\lambda_i &= (i-1)\frac{\lambda_M}{M-1} \text{ for } i=1,\ldots,M,
\end{align*}
and plot the corresponding time-scaled \ac{MI} in colored solid lines with markers. We note that the time-scaled \ac{MI} is very close to the capacity curve until it saturates. Increasing the modulation order \(M\) shows significant increase in the time-scaled \ac{MI}. 
For comparison purposes, we also plot the time-scaled \ac{MI} for a system with fixed symbol duration and conventional uniform distribution on a linearly spaced constellation as in~\cite{shevchenko2018tcom}. We observe that the rate saturates at very low values and that increasing the modulation order \(M\) shows only slight improvement.

\begin{figure}[tb]
  \centering
		\includegraphics{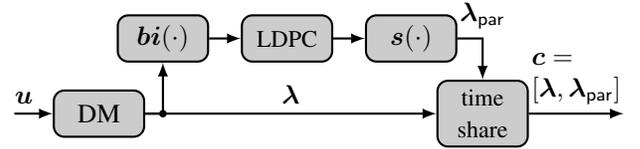} 
		\caption{Block diagram of the \ac{PES} scheme.}
		\label{fig:shaping:system_block_diagram_pes}
\end{figure}
\begin{figure*}[tb]
  \centering
		\includegraphics{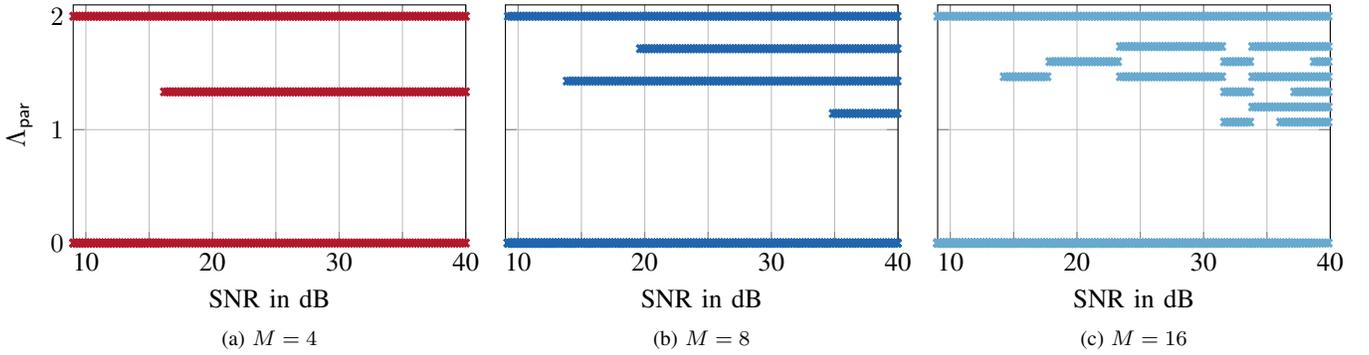}
		\caption{Resulting constellations for the parity symbols for different \acp{SNR}. Note that the highest and the lowest eigenvalue is always occupied for every modulation order.}
		\label{fig:shaping:geometric_parity_constellation}
\end{figure*}
\section{Probabilistic Eigenvalue Shaping}
\label{sec:shaping_time_share}
In the previous section, we observed a significant gap between the time-scaled \ac{MI} of the system in~\cite{shevchenko2018tcom} and the capacity. This gap is referred to as shaping gap. In order to close it, we propose a \ac{PES} system as shown in \internalFig{fig:shaping:system_block_diagram_pes}, inspired by \ac{PAS}~\cite{bocherer2015_bw_efficient_rate_matched_ldpc}.

In the \ac{PAS} scheme, the sequence of uniformly distributed data bits is mapped to a sequence of positive amplitudes distributed half Gaussian by a \ac{DM}. The binary image of this sequence is encoded by a systematic \ac{FEC} code, resulting in uniformly distributed parity bits, which are then used to map the sequence of half Gaussian distributed symbols to a stream of Gaussian distributed symbols. 

As the capacity-achieving distribution \(p_X^*(\lambda)\) is not symmetric, \ac{PAS} cannot be directly applied here.
However, in order to keep the benefits of \ac{PAS}, we wish to apply the \ac{DM} before the \ac{FEC}. We describe \ac{PES} in the following with reference to \internalFig{fig:shaping:system_block_diagram_pes}.
The binary data sequence \(\bm{u}\) of length \(k_\mathsf{s}\) bits is mapped by the \ac{DM} to a sequence of eigenvalues \(\bm{\lambda}\in\Lambda^{n_\mathsf{s}}\) of length \(n_\mathsf{s}\) distributed according to \(p_X^*(\lambda)\). 
The \ac{CCDM} can be used for that purpose~\cite{schulte2016_ccdm}. It is asymptotically optimal as its rate \(R_\mathsf{s}\) approaches the entropy of the desired channel input \(X\),
\begin{align*}
	R_\mathsf{s} &=\frac{k_\mathsf{s}}{n_\mathsf{s}} \rightarrow \entropy{X} \text{  as  } n_\mathsf{s} \rightarrow\infty.
\end{align*}
For large block sizes, the gap between \(R_\mathsf{s}\) and \(\entropy{X}\) is sufficiently small and can be neglected. Note that some of the possible eigenvalues may occur with probability zero.

We consider the modulation order \(M\) to be a power of two such that we can define its binary image. The binary image of \(\bm{\lambda}\), \(\bm{bi}(\bm{\lambda} )\), is then encoded by a systematic encoder with information block length \(k_\mathsf{c}\), code length \(n_\mathsf{c}\), and rate \(R_\mathsf{c}=\frac{k_\mathsf{c}}{n_\mathsf{c}}\). The code is denoted by \(\mathcal{C}\), with \(|\mathcal{C}|=2^{k_\mathsf{c}}\). 
The parity bits at the output of the encoder are mapped to a sequence of eigenvalues \(\bm{\lambda}_\mathsf{par}\in \Lambda_\mathsf{par}\) with modulation order \(M_\mathsf{par}=|\Lambda_\mathsf{par}|\) and \(\Lambda_\mathsf{par}\subseteq\Lambda\) by the block \(\bm{s}(\cdot )\) in \internalFig{fig:shaping:system_block_diagram_pes} such that they are uniformly distributed.

Assuming that a high code rate \(R_\mathsf{c}\) is used, we accept a small penalty with respect to the optimal channel input distribution and transmit  \(\bm{\lambda}\) and \(\bm{\lambda}_\mathsf{par}\) via time-sharing.
The major difference of \ac{PES} compared to \ac{PAS} is the fact that the channel input distribution is not the optimal distribution due to the time-sharing with the sequence \(\bm{\lambda}_\mathsf{par}\). Consequently, this causes a performance degradation. However, \ac{PES} is highly flexible as the spectral efficiency can be adapted by the \ac{DM} and the code rate \(R_\mathsf{c}\), and a single code can be used. Note that every eigenvalue is protected by the code as \ac{FEC} is performed after the \ac{DM} and decoding and demapping can be performed independently. Thus, \Ac{PES} shares these advantages with \ac{PAS}.

We wish for a high code rate \(R_\mathsf{c}\) to keep the performance degradation due to the time-sharing low. More precisely, we wish to maximize the number of symbols distributed according to \(p_X^*(\lambda)\). The ratio between information symbols and coded symbols, denoted by \(R_\mathsf{ts}\), is an indication for the expected performance degradation, 
\begin{align}
R_\mathsf{ts} &= \frac{\frac{n_\mathsf{c} R_\mathsf{c}}{\log_2(M)}}{\frac{n_\mathsf{c}R_\mathsf{c}}{\log_2(M)} + \frac{n_\mathsf{c}(1-R_\mathsf{c})}{\log_2(M_\mathsf{par})}}\nonumber\\
&= \frac{R_\mathsf{c}\log_2(M_\mathsf{par})}{\log_2(M)(1-R_\mathsf{c})+R_\mathsf{c}\log_2(M_\mathsf{par})}.
\label{eqn:shaping:R_ts}
\end{align}

\subsection{Parity symbols}
The parity symbols at the output of the \ac{FEC} code encoder are uniformly distributed. In \internalFig{fig:shaping:mi_comparison}, we observed that \ac{OOK} with uniform signaling, i.e., \(\Lambda_\mathsf{par}=\{\lambda_1, \lambda_M\}\) and \(M_\mathsf{par}=2\),  is optimal for low \ac{SNR} as it achieves capacity and performs reasonably well for high \ac{SNR}. However, we note from \internalFig{fig:shaping:mi_comparison} that for a higher order modulation, even with uniform signaling, higher rates are possible. Hence, here we consider a scenario where \(M_\mathsf{par}>2\).
We further increase the rate by only using a subset of \(\Lambda\) and by picking the eigenvalues such that they are not uniformly spaced.
\begin{example}
Consider the information symbol alphabet \(\Lambda=\{\lambda_1,\ldots,\lambda_8\}\) with \(M=8\). For the \(\Lambda_\mathsf{par}\), we could pick \(\Lambda_\mathsf{par}=\{\lambda_1, \lambda_6, \lambda_7,\lambda_8\}\) with \(p_X(\lambda)=\{0.25, 0.25, 0.25, 0.25\}\) and \(M_\mathsf{par}=4\).
\end{example}
To find the function \(\bm{s}(\cdot )\) that maps the parity symbols onto \(\lambda\in\Lambda_\mathsf{par}\), we use a greedy algorithm as described in \internalAlgorithm{lst:algo_geometric_shaping}. It starts with \ac{OOK}, i.e., \(\Lambda_\mathsf{par}=\{\lambda_1, \lambda_M\}\). For each of the remaining symbols \(\lambda\in\Lambda\setminus\Lambda_\mathsf{par}\), it calculates the time-scaled \ac{MI} of \(\lambda\cup\Lambda_\mathsf{par}\),
finds the symbol \(\lambda\) for which the time-scaled \ac{MI} of \(\lambda\in\Lambda\setminus\Lambda_\mathsf{par}\) is maximized, and adds it to \(\Lambda_\mathsf{par}\). All symbols with a greater or equal imaginary part than \(\lambda\) are removed, i.e., the eigenvalues  \(\{\lambda'\in\Lambda : \Im{\lambda'}\geq\Im{\lambda}\}\) are removed. This process is repeated until there are no symbols left. We then choose the set of symbols that gives the highest time-scaled \ac{MI} as \(\Lambda_\mathsf{par}\).
We note that this procedure does not guarantee an optimal solution. However, for \(M=\{4,8\}\) an exhaustive search gives the same result as that of \internalAlgorithm{lst:algo_geometric_shaping}.

\begin{algorithm}
 \caption{Algorithm to calculate the signal points for the parity symbols. With a slight abuse of notation, we denote the time-scaled \ac{MI} of a set \(\Lambda_\mathsf{par}\) by \(\Istar{\Lambda_\mathsf{par}}\). We assume the symbols in the set to be uniformly distributed.}
 \label{lst:algo_geometric_shaping}
 \begin{algorithmic}[1]
 \renewcommand{\algorithmicrequire}{\textbf{Input:}}
 \renewcommand{\algorithmicensure}{\textbf{Output:}}
 \REQUIRE Constellation \(\Lambda\)
 \ENSURE  Constellation \(\Lambda_\mathsf{par}\)
  \STATE \(\Lambda_\mathrm{placed} = \{\lambda_1, \lambda_M\}\)
	\STATE \(\Lambda_\mathsf{par} = \{\lambda_1, \lambda_M\}\)

	\STATE \(\Lambda_\mathrm{not\,placed} = \Lambda \setminus \Lambda_\mathsf{par}\)
  \WHILE {\(\Lambda_\mathrm{not\,placed} \neq \emptyset\)}
	\FORALL {\(\lambda_i\in\Lambda_\mathrm{not\,placed}\)}
	\STATE Calculate \(\Istar{\Lambda_\mathrm{placed} \cup \lambda_i}\)
	\ENDFOR
	\STATE \(\lambda_\mathrm{max} := \arg\max \Istar{\cdot}\)
	\STATE \(\Lambda_\mathrm{placed} = \Lambda_\mathrm{placed} \cup \lambda_\mathrm{max}\)
	\IF{\(\Istar{\Lambda_\mathrm{placed}} > \Istar{\Lambda_\mathsf{par}}\)}
	\STATE \(\Lambda_\mathsf{par} = \Lambda_\mathrm{placed}\)
	\ENDIF
  \STATE \(\Lambda_\mathrm{not\,placed} = \Lambda_\mathrm{not\,placed}\setminus\{\lambda:\lambda\in\Lambda_\mathrm{not\,placed}, \Im{\lambda}\geq\Im{\lambda_\mathrm{max}}\}\)
  \ENDWHILE
	
 \RETURN \(\Lambda_\mathsf{par}\) 
 \end{algorithmic} 
 \end{algorithm}
In \internalFig{fig:shaping:geometric_parity_constellation}, we show \(\Lambda_\mathsf{par}\) for different modulation orders and \acp{SNR}. For \(M=4\), we note that for low \ac{SNR} \ac{OOK} gives the best result. Increasing the \ac{SNR} results in a third level being added. The same behavior is observed for \(M=8\). Compared to \(M=4\), the third level is introduced at a slightly lower \ac{SNR}. This results from the fact that for \(M=8\), different constellation points are available.
For \(M=16\), we note that again a third level appears when increasing the \ac{SNR}. When further increasing it, this third level moves to an eigenvalue with larger imaginary part and consequently a fourth level at an eigenvalue with lower imaginary part appears. This behavior can be observed repeatedly. To map the binary parity bits to the constellation points, we require \(M_\mathsf{par}\) to be a power of two. As this is not always the case (see \internalFig{fig:shaping:geometric_parity_constellation}), we pick the largest power of two that is smaller or equal than the number of constellation points given by \internalAlgorithm{lst:algo_geometric_shaping}.

\subsection{Achievable Rate of Probabilistic Eigenvalue Shaping}
\label{sec:achievable_rate}
To characterize the performance of \ac{PES}, we derive the achievable rate of \ac{PES}, denoted by \(\mathsf{R}_\mathsf{ps}\).
We assume that the channel is memoryless and that the decoder performs bit-metric decoding.
\begin{theorem}
The achievable rate of \ac{PES} is
\begin{align}
\mathsf{R}_\mathsf{ps} &= R_\mathsf{ts}\left(\entropy{X}  - \sum\limits_{i=1}^{m}\entropy{X^{\mathsf{B}}_i|Y^{\mathsf{B}}_i} \right) \nonumber\\ 
&\quad + \left(1-R_\mathsf{ts}\right)\left(m_\mathsf{par} - \sum\limits_{i=1}^{m_\mathsf{par}}\entropy{X^{\mathsf{B}}_{\mathsf{par},i}|Y^{\mathsf{B}}_{\mathsf{par},i}}\right).\label{eqn:air}
\end{align}
\end{theorem}
\begin{proof}
The achievable rate for \ac{PAS} has been derived in~\cite{bocherer2017_achievable_rate_ps}. For a system employing time-sharing, the resulting achievable rate is the average of the achievable rate of the two transmission schemes.
\end{proof}
\begin{figure}[tb]
  \centering
	\includegraphics{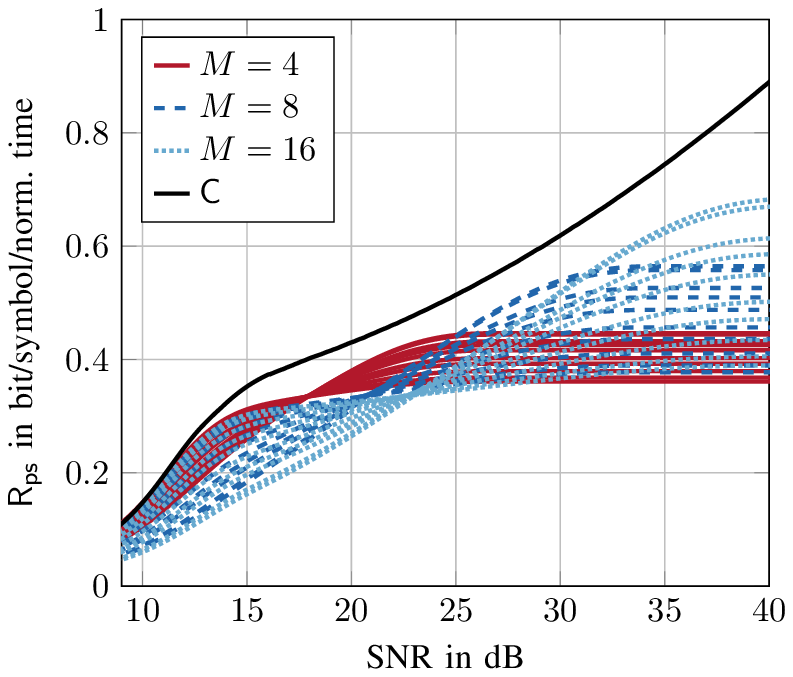}
	\caption{Achievable rates for different code rates with \(\Lambda_\mathsf{par}\) according to \internalAlgorithm{lst:algo_geometric_shaping} for a cutoff parameter \(\delta=0.005\).}
		\label{fig:shaping:achievable_rate_air}
\end{figure}

In \internalFig{fig:shaping:achievable_rate_air}, we plot the capacity and the achievable rate \internalEq{eqn:air} for different code rates \(R_\mathsf{c}=\{\)\(1/4\), \(1/3\), \(2/5\), \(1/2\), \(3/5\), \(2/3\), \(3/4\), \(4/5\), \(5/6\), \(8/9\), \(9/10\}\)
and modulation orders for a cutoff parameter \(\delta=0.005\). \(\Lambda_\mathsf{par}\) and hence \(M_\mathsf{par}\) are chosen according to the results of \internalAlgorithm{lst:algo_geometric_shaping}. For each modulation order, we notice that the curves cross at a certain \ac{SNR}. For \acp{SNR} below this point, the lowest code rate (corresponding to the highest curve) gives the best performance whereas for \acp{SNR} above this point, the highest code rate (corresponding to the highest curve) gives the best performance. We note the influence of time-sharing, which results in a gap between the achievable rate and capacity. The gap increases for lower code rates \(R_\mathsf{c}\) as the channel input distribution deviates more from the optimal one.

\section{Numerical Evaluation}
\label{sec:performance}
In this section, we evaluate the performance of the \ac{PES} scheme via discrete-time Monte-Carlo and \ac{SSFM} simulations.
For the mapping \(\bm{bi}(\cdot)\) (see \internalFig{fig:shaping:system_block_diagram_pes}), we use Gray labeling. Also, for the \ac{FEC}, we use the binary \ac{LDPC} codes of the \acs{DVBS2} standard with code length \(n_\mathsf{c}=64800\) and code rates \(R_\mathsf{c}=\{\)\(1/4\), \(1/3\), \(2/5\), \(1/2\), \(3/5\), \(2/3\), \(3/4\), \(4/5\), \(5/6\), \(8/9\), \(9/10\}\). For the parity symbols, we use the constellation arising from \internalAlgorithm{lst:algo_geometric_shaping}, depicted in \internalFig{fig:shaping:geometric_parity_constellation}. 

\subsection{Detection}
\label{sec:detection}
For the \ac{SSFM} simulation, we simulate a continuous signal and hence, we require a detector. We use the following method to deal with the variable pulse durations: We set a threshold \(\theta\) sufficiently higher than the noise. Once the magnitude of the signal rises above \(\theta\), we save the time as \(t_\mathsf{start}\) and when the magnitude of the signal falls below \(\theta\), we save the time as \(t_\mathsf{end}\). We then extend the interval bounded by \(t_\mathsf{start}\) and \(t_\mathsf{end}\), i.e., \(\tilde{t}_\mathsf{start} = t_\mathsf{start} - \delta_t\) and \(\tilde{t}_\mathsf{end} = t_\mathsf{end} + \delta_t\). Calculating the \ac{NFT} over the interval \([\tilde{t}_\mathsf{start}, \tilde{t}_\mathsf{end}]\) using the spectral method~\cite[Part II, Section IV]{yousefi2014nftI-III} and only considering the imaginary part of the discrete eigenvalue gives the received symbol \(y\). This approach requires that the \ac{SNR} is sufficiently high. As the model has the same requirement due to the perturbation approach, this requirement is fulfilled.

It may happen that due to noise, a received pulse never rises above the threshold \(\theta\). In this case, the shortest duration is assumed (i.e., the duration of the pulse with amplitude zero). This scenario can be avoided by choosing the threshold sufficiently lower than the lowest amplitude. Furthermore, due to the shape of the capacity-achieving distribution, lower amplitudes are less likely, hence preventing this scenario.

To find the best threshold, we tested the performance for different values of \(\theta\) and found that the performance of a threshold at \(75 \%\) of the lowest non-zero amplitude of the constellation works best. We observed that small deviations of the threshold do not affect the performance significantly whereas setting the threshold too high (missing symbols with low amplitude) or to low (detecting a symbol where there is none) leads to performance degradation.
Furthermore, we assume synchronization sequences spread sufficiently far apart in order not to impact the rate. We assume synchronization to be ideal such that it is guaranteed that error propagation is limited.

\ctable[
cap     = Simulation parameters,
caption = Simulation parameters.,
label   = tab:fiber_param,
pos     = tb,
width = \columnwidth,
doinside = \small
]{llr}{
}{                                                          \FL
Span length &  \(l_\mathrm{span}\) & \(\SI{80}{\kilo\meter}\)   \NN
Second order dispersion & \(\beta_2\) & \(\SI{-21.137}{\square\pico\second\per\kilo\meter}\)\NN
Nonlinearity parameter & \(\gamma\) & \(\SI{1.4}{\per\watt\per\kilo\meter}\)\NN
Attenuation & \(\alpha\) & \(\SI{0.2}{\decibel\per\kilo\meter}\)\NN
Shortest pulse & \(T_\mathsf{short}\) & \(\SI{0.83}{\nano\second}\)\NN
Longest pulse & \(T_\mathsf{long}\) & \(\SI{1.3}{\nano\second}\)\NN
Bandwidth & \(B\) & \(\SI{1.2}{\giga\hertz}\)\NN
Avg. transmit power & \(P\) & \(\SI{-4.86}{\dBm}\)\NN
Cutoff parameter & \(\delta\) & \(0.005\)\LL
}

\subsection{Numerical Results}
We perform Monte-Carlo simulations of the discrete-time model \internalEq{eqn:pdf_yx_sac} and show the results in \internalFig{fig:performance:mi_monte_carlo_perf_geometric}, where we plot the transmission rate at a \ac{BER} of \(10^{-5}\) for \(M=4,8\) and \(16\). The highest transmission rate for each modulation order corresponds to the highest code rate \(R_\mathsf{c}\). We notice that the gap to capacity for \(M=4\) is smaller than for \(M=8\) and \(M=16\).
If we consider \(\Delta M = M - M_\mathsf{par}\), i.e., the difference of the modulation order of \(\Lambda\) and \(\Lambda_\mathsf{par}\), we note that for a low \(M\), \(\Delta M\) is low was well. For example, for \(M=4\), \(\Delta M \leq 2\). Hence, the rate loss due to time-sharing is small. For \(M=16\), the gap to capacity is smaller than for \(M=8\). Considering the relevant \ac{SNR} range, we note that \(\Delta M\) is smaller for \(M=16\) than for \(M=8\) and thus explaining the smaller rate loss.

We also simulated the transmission over a fiber using \ac{SSFM} simulations transmitting a train of solitons. We consider a \ac{SMF} with parameters as in \internalTab{tab:fiber_param} and two different amplification schemes, distributed Raman amplification and lumped amplification using \acp{EDFA}.
For both schemes, the peak power constraint is chosen such that the effect of the \acp{EDFA} can be neglected, i.e., \(\lambda_\mathsf{max}=2\jmath\). We employ the detection schemes as described in \internalLink{sec:detection} and choose the cutoff-parameter \(\delta=0.005\), i.e., \(99.5\,\%\) of the energy is contained in the pulse, for which the condition \internalEq{eqn:guard_time_condition} is fulfilled. This then leads to a similar cutoff parameter as in~\cite{shevchenko2018tcom}.
For each modulation order \(M=\{4,8,16\}\), we determine the furthest distance over which we achieve a \ac{BER} of less than \(10^{-5}\) and consider the rate gain compared to an unshaped system as in~\cite{shevchenko2018tcom}.
This results for \(M=\{4,8,16\}\) in  transmission over \(\SI{3200}{\kilo\meter}\), \(\SI{3040}{\kilo\meter}\), and \(\SI{2960}{\kilo\meter}\) at a rate gain of \(\SI{20}{\percent}\), \(\SI{26}{\percent}\), and \(\SI{95}{\percent}\), respectively. The results do not differ for distributed and lumped amplification as this is ensured by the peak power constraint.

\section{Conclusion}
\label{sec:conclusion}
In this paper, we presented a probabilistic shaping scheme for an \ac{NFT}-based transmission system embedding information in the imaginary part of the discrete spectrum. It shapes the information symbols according to the capacity-achieving distribution and transmits them via time-sharing together with the uniformly distributed, suitably modulated parity symbols. 
We exploited the fact that the pulses of the signal in the time domain are of unequal length to improve the data rate compared to~\cite{shevchenko2018tcom}. We used the time-scaled \ac{MI} and derived the capacity-achieving distribution in closed form for the noiseless case and numerically in the general case.
We showed that \acl{PES} significantly improves the performance of an \ac{NFT}-based transmission scheme, and can almost double the data rate. As a possible extension of our work, the continuous spectrum can be used to increase the spectral efficiency~\cite{aref2016ecoc}.

\begin{figure}[tb]
  \centering
		\includegraphics{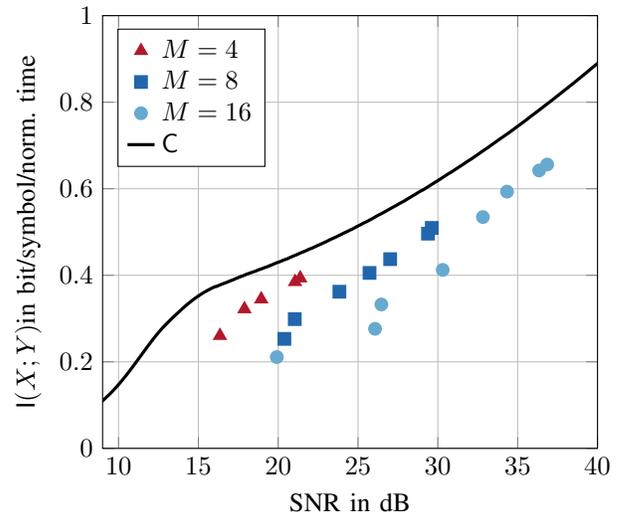}
		\caption{Performance of time sharing with parity symbols according to \internalAlgorithm{lst:algo_geometric_shaping}. The rate points correspond to a performance at \(\text{BER}=10^{-5}\). The highest transmission rate for each modulation order corresponds to the highest code rate \(R_\mathsf{c}\).}
		\label{fig:performance:mi_monte_carlo_perf_geometric}
\end{figure}

\section*{Acknowledgments}
The authors would like to thank the anonymous reviewers for their feedback and comments which helped to improve this paper significantly. Especially, we would like to acknowledge one of the reviewers for proposing an elegant way to prove \internalLemma{lemma:noiseless}, which is included in this paper.

\ifCLASSOPTIONcaptionsoff
  \newpage
\fi


\vfill

\begin{thebibliography}{10}
\providecommand{\url}[1]{#1}
\csname url@samestyle\endcsname
\providecommand{\newblock}{\relax}
\providecommand{\bibinfo}[2]{#2}
\providecommand{\BIBentrySTDinterwordspacing}{\spaceskip=0pt\relax}
\providecommand{\BIBentryALTinterwordstretchfactor}{4}
\providecommand{\BIBentryALTinterwordspacing}{\spaceskip=\fontdimen2\font plus
\BIBentryALTinterwordstretchfactor\fontdimen3\font minus
  \fontdimen4\font\relax}
\providecommand{\BIBforeignlanguage}[2]{{%
\expandafter\ifx\csname l@#1\endcsname\relax
\typeout{** WARNING: IEEEtran.bst: No hyphenation pattern has been}%
\typeout{** loaded for the language `#1'. Using the pattern for}%
\typeout{** the default language instead.}%
\else
\language=\csname l@#1\endcsname
\fi
#2}}
\providecommand{\BIBdecl}{\relax}
\BIBdecl

\bibitem{yousefi2014nftI-III}
M.~I. Yousefi and F.~R. Kschischang, ``Information transmission using the
  nonlinear fourier transform, part {I}-{III},'' \emph{{IEEE} Trans. Inf.
  Theory}, vol.~60, no.~7, pp. 4312--4369, Jul. 2014.

\bibitem{dong2015}
Z.~Dong, S.~Hari, T.~Gui, K.~Zhong, M.~I. Yousefi, C.~Lu, P.~K.~A. Wai, F.~R.
  Kschischang, and A.~P.~T. Lau, ``Nonlinear frequency division multiplexed
  transmissions based on {NFT},'' \emph{{IEEE} Photon. Technol. Lett.},
  vol.~27, no.~15, pp. 1621--1623, Aug. 2015.

\bibitem{aref2015ecoc}
V.~Aref, H.~B{\"u}low, K.~Schuh, and W.~Idler, ``Experimental demonstration of
  nonlinear frequency division multiplexed transmission,'' in \emph{Proc. 41st
  Eur. Conf. Opt. Commun. ({ECOC})}, Valencia, Spain, Sep. 2015, pp. 1--3.

\bibitem{aref2016ecoc}
V.~Aref, S.~T. Le, and H.~B{\"u}low, ``Demonstration of fully nonlinear
  spectrum modulated system in the highly nonlinear optical transmission
  regime,'' in \emph{Proc. 42nd Eur. Conf. Opt. Commun. ({ECOC})},
  D{\"u}sseldorf, Germany, Sep. 2016, pp. 1--3.

\bibitem{geisler2016ecoc}
A.~Geisler and C.~Schaeffer, ``Experimental nonlinear frequency division
  multiplexed transmission using eigenvalues with symmetric real part,'' in
  \emph{Proc. 42nd Eur. Conf. Opt. Commun. ({ECOC})}, D{\"u}sseldorf, Germany,
  Sep. 2016, pp. 1--3.

\bibitem{hari2016}
S.~Hari, M.~I. Yousefi, and F.~R. Kschischang, ``Multieigenvalue
  communication,'' \emph{J. Lightw. Technol.}, vol.~34, no.~13, pp. 3110--3117,
  Jul. 2016.

\bibitem{shevchenko2018tcom}
N.~A. Shevchenko, S.~A. Derevyanko, J.~E. Prilepsky, A.~Alvarado, P.~Bayvel,
  and S.~K. Turitsyn, ``Capacity lower bounds of the noncentral chi-channel
  with applications to soliton amplitude modulation,'' \emph{{IEEE} Trans.
  Commun.}, to appear.

\bibitem{forney1984}
G.~D. Forney, R.~Gallager, G.~Lang, F.~Longstaff, and S.~Qureshi, ``Efficient
  modulation for band-limited channels,'' \emph{{IEEE} J. Sel. Areas Commun.},
  vol.~2, no.~5, pp. 632--647, Sep. 1984.

\bibitem{sun1993_geometric_shaping}
F.-W. Sun and H.~C.~A. van Tilborg, ``Approaching capacity by equiprobable
  signaling on the {G}aussian channel,'' \emph{{IEEE} Trans. Inf. Theory},
  vol.~39, no.~5, pp. 1714--1716, Sep. 1993.

\bibitem{bocherer2015_bw_efficient_rate_matched_ldpc}
G.~B{\"o}cherer, F.~Steiner, and P.~Schulte, ``Bandwidth efficient and
  rate-matched low-density parity-check coded modulation,'' \emph{{IEEE} Trans.
  Commun.}, vol.~63, no.~12, pp. 4651--4665, Dec. 2015.

\bibitem{derevyanko2005}
S.~A. Derevyanko, S.~K. Turitsyn, and D.~A. Yakushev, ``Fokker-planck equation
  approach to the description of soliton statistics in optical fiber
  transmission systems,'' \emph{J. Opt. Soc. Am. B}, vol.~22, no.~4, pp.
  743--752, Apr. 2005.

\bibitem{zafrulla2003}
M.~Zafruullah, M.~Waris, and M.~K. Islam, ``Simulation and design of {EDFA}s
  for long-haul soliton based communication systems,'' in \emph{Proc.
  Asia-Pacific Conf. Commun. ({APCC})}, Penang, Malaysia, Sep. 2003.

\bibitem{verdu1990}
S.~Verd{\'u}, ``On channel capacity per unit cost,'' \emph{{IEEE} Trans. Inf.
  Theory}, vol.~36, no.~5, pp. 1019--1030, Sep. 1990.

\bibitem{stancu-minasian_frac_programming}
I.~M. Stancu-Minasian, \emph{Fractional Programming}, 1st~ed.\hskip 1em plus
  0.5em minus 0.4em\relax Dordrecht, The Netherlands: Kluwer Academic
  Publishers, 1997.

\bibitem{cover_inf_theory}
T.~M. Cover and J.~A. Thomas, \emph{Elements of Information Theory},
  2nd~ed.\hskip 1em plus 0.5em minus 0.4em\relax Hoboken, NJ, USA: Wiley, 2006.

\bibitem{schulte2016_ccdm}
P.~Schulte and G.~B{\"o}cherer, ``Constant composition distribution matching,''
  \emph{{IEEE} Trans. Inf. Theory}, vol.~62, no.~1, pp. 430--434, Jan. 2016.

\bibitem{bocherer2017_achievable_rate_ps}
\BIBentryALTinterwordspacing
G.~B{\"o}cherer, ``Achievable rates for probabilistic shaping,'' \emph{ArXiv
  e-prints}, Jul. 2017. [Online]. Available:
  \url{https://arxiv.org/abs/1707.01134}
\BIBentrySTDinterwordspacing

\end{thebibliography}
\end{document}